\documentclass[submission,copyright,creativecommons]{eptcs}
\providecommand{\event}{GandALF 2012} 
\usepackage{pstricks,pst-node,pst-tree}
\usepackage{amssymb}
\usepackage{amsmath}
\usepackage{graphicx}
\usepackage{amsthm}
\usepackage{ifthen}

\newtheorem{theorem}{Theorem}[section]
\newtheorem{prop}[theorem]{Proposition}
\newtheorem{lemma}[theorem]{Lemma}
\theoremstyle{definition}
\newtheorem{definition}[theorem]{Definition}
\newtheorem{example}[theorem]{Example}

\newcommand{\ign}[1]{}
\newcommand{\Aa}{{\cal A}}

\newcommand{\Dd}{{\cal D}}

\newcommand{\Pp}{{\cal P}}

\newcommand{\Tt}{{\cal T}}
\newcommand{\DV}{\mathbb{D}}
\newcommand{\idv}{\mathsf{idv}}
\newcommand{\nodata}{\mathsf{no\textrm{-}data}}
\newcommand{\type}{\mathsf{type}}
\newcommand{\origin}{\mathsf{src}}
\newcommand{\pop}[1]{\mathsf{pop}^{#1}}
\newcommand{\push}[1]{\mathsf{push}^{#1}}
\newcommand{\collapse}[1]{\mathsf{collapse}^{#1}}

\newcommand{\lab}[1]{\label{#1}} 
\newcommand{\subrun}[3]{#1{\restriction}_{#2,#3}}
\newcommand{\Rz}[1]{R\!\langle #1\rangle}
\newcommand{\Iff}{\Leftrightarrow}
\newcommand{\przyp}[1]{{\it #1}}

\title{Higher-Order Pushdown Systems with Data}
\author{Pawe\l\ Parys\thanks{Work partially supported by the Polish Ministry of Science grant nr N N206 567840.}
\institute{University of Warsaw\\
Warszawa, Poland}
\email{parys@mimuw.edu.pl}}
\def\titlerunning{Higher-Order Pushdown Automata with Data}
\def\authorrunning{Pawe\l\ Parys}

\begin{document}

\maketitle

\begin{abstract}
	We propose a new extension of higher-order pushdown automata, which allows to use an infinite alphabet.
	The new automata recognize languages of data words (instead of normal words), which beside each its letter from a finite alphabet have a data value from an infinite alphabet.
	Those data values can be loaded to the stack of the automaton, and later compared with some farther data values on the input.
	Our main purpose for introducing these automata is that they may help in analyzing normal automata (without data).
	As an example, we give a proof that deterministic automata with collapse can recognize more languages than deterministic automata without collapse.
	This proof is simpler than in the no-data case.
	We also state a hypothesis how the new automaton model can be related to the original model of higher-order pushdown automata.
\end{abstract}

\section{Introduction}

Already in the 70's, Maslov (\cite{Mas74,Mas76}) generalized the concept of pushdown automata to higher-order pushdown automata 
and studied such devices as acceptors of string languages.
In the last decade, renewed interest in these automata has arisen. 
They are now studied also as generators of graphs and trees.
Knapik et al.\ \cite{easy-trees} showed that the class of trees generated by deterministic level-$n$ pushdown automata coincides 
with the class of trees generated by \emph{safe} level-$n$ recursion schemes (safety is a  syntactic restriction on the recursion scheme),
and Caucal \cite{Caucal02} gave another characterization: trees on level $n+1$ are obtained from trees on level $n$
by an MSO-interpretation of a graph, followed by application of unfolding.
Carayol and W\"ohrle \cite{cawo03} studied the $\varepsilon$-closures of configuration graphs of level-$n$ pushdown automata
and proved that these graphs are exactly the graphs in the $n$-th level of the Caucal hierarchy.
In order to deal with non-safe schemes,
Hague et al.~\cite{collapsible} extended the model of level-$n$ pushdown automata to level-$n$ collapsible pushdown automata by
introducing a new stack operation called collapse. 
They showed that the trees generated by such automata coincide exactly 
with the class of trees generated by all higher-order recursion schemes and this correspondence is level-by-level. 
Nevertheless, for a long time it was an open question whether safety implies a semantical restriction to recursion schemes.
It was proven recently \cite{parys-panic-new}, and the proof is technically rather complicated.

We extend the definition of higher-order pushdown automata by allowing infinite alphabets, both on input and on stack.
More precisely, every symbol consists of two parts: a label from a finite alphabet, and a data value from an infinite alphabet.
The part from the finite alphabet is treated as normally.
On the other hand all the data values are assumed to be equivalent; the automaton can only test whether two data values are equal 
(data values from the input can be stored on the stack, and afterwards compared with some future data values on the input).
Words over such a two-part alphabet are called data words.
We consider both collapsible and non-collapsible variant of the automata.

Data words are widely considered, for example in the context of XML processing.
In XML documents we have tag names and attribute names (which are from a finite alphabet),
but we also have a text data; typically we don't want to analyze the internal structure of such a text, we just want to compare texts appearing indifferent places in the document 
(or to compare them with some constants).
Several automata models and logics were already considered for data words, see \cite{SegoufinData} for a survey. 
Some of the automata models are generalizations of finite automata to data words \cite{KaminskiF94}, other \cite{data_aut, class_aut} are designed to cover some logics, like FO$^2$ or XPath.
A model of pushdown automata (of level $1$) using an infinite alphabet (very similar to our model) was considered in \cite{ChengK98}.
(Historically the first models of pushdown automata over infinite alphabet were presented in \cite{autebert} and \cite{Idt84};
they simply allow infinitely many rules and infinitely many transition functions, one for each letter in the alphabet).
We are not aware of any works which consider infinite alphabets for higher-order pushdown automata.

The main technical content of this paper is a proof that deterministic collapsible pushdown automata with data recognize more languages than non-collapsible ones, as stated below.

\begin{theorem}\lab{thm:main}
	There exists a data language recognized by a deterministic collapsible pushdown automaton with data of the second level, 
	which is not recognized by any deterministic higher-order pushdown automaton with data (without collapse) of any level.
\end{theorem}

The proof of this theorem is extracted from the analogous proof for automata without data, contained in \cite{parys-panic-new}.
Although we consider a more complicated model of automata, the proof is in fact simpler.
The reason is that when having data values, we can precisely trace where this data value is stored in the stack.
We just need to analyze how this places are related for different data values, hence in which order the data values can be recovered from the stack.
The idea for automata without data is similar, but its realization is much more complicated.
Instead of considering data values, we consider blocks containing the same letter repeated very many times.
We look in which place of the stack the automaton ``stores'' the number of these repetitions.
However now this place cannot be easily identified.
Moreover, it is difficult to change the number of repetitions in one place, so that the whole run stays similar.
This problems disappear when we have data values.
Thus our paper can be rather seen as an illustration of the original proof.

We also state a hypothesis how the new classes of automata are related to the original one.
Namely, let us consider the (already mentioned) encoding of data words into normal words, in which a data value is represented as a number of repetitions of a symbol.
Of course every data word has many encodings, as a data value may be mapped to an arbitrary number of repetitions.
Assume that we have an automaton (for normal words) which either accepts all representations of a data word, or none of them.
Our hypothesis would be that then we can create an automaton with data, which recognizes the corresponding language of data words.
One can easily see that the opposite implication holds: an automaton with data can be converted to a normal automaton recognizing all representations of data words in the language
(the number of repetitions of a symbol on the input is stored on the stack as a number of repetitions of some stack symbol).

If this hypothesis would be proved to be true, our Theorem \ref{thm:main} would imply the original theorem of \cite{parys-panic-new} 
that collapsible pushdown automata (without data) recognize more languages than non-collapsible ones
(the representation of the data word language from Theorem \ref{thm:main} would give a correct example language for normal words).

A next goal would be to prove analogous result for nondeterministic automata, i.e.~that nondeterministic collapsible pushdown automata with data recognize more languages than non-collapsible ones.
Notice that this is not the case for level $2$ \cite{AehligMO05}, but is believed to be true for higher levels.
It looks likely that this can be shown for automata with data, much easier than for automata without data.

\paragraph*{Acknowledgment.}
We thank M. Boja\'nczyk for proposing the idea of introducing data values to higher-order pushdown automata.

\section{Preliminaries}

In this section we define higher-order pushdown automata with data.

Let us fix any infinite set of data values $\DV$, e.g.~$\DV=\mathbb{N}$.
For a finite alphabet $A$, a \emph{data word} over $A$ is just a word over $A\times\DV$.
For a data word $w$, by $\pi(w)$ we denote the (normal) word over $A$ obtained from $w$ by projecting every letter to its first coordinate.

For any alphabet $\Gamma$ (of stack symbols) we define a \emph{$k$-th level stack} (a \emph{$k$-stack} for short) 
as an element of the following set $\Gamma^k_*$:
\begin{align*}
	&\Gamma^0_*=\Gamma\times(\DV\cup\{\nodata\}),\\
	&\Gamma^k_*=(\Gamma^{k-1}_*)^*\qquad\mbox{for }1\leq k\leq n.
\end{align*}
In other words, a $0$-stack contains a symbol from $\Gamma$ and possibly a data value (or $\nodata$ if no data value is stored), and a $k$-stack for $1\leq k\leq n$ is a (possibly empty) sequence of $(k-1)$-stacks.
Top of a stack is on the right.
The \emph{size} of a $k$-stack is just the number of $(k-1)$-stacks it contains.
For any $s^k\in\Gamma^k_*$ and $s^{k-1}\in\Gamma^{k-1}_*$ we write $s^k:s^{k-1}$ for the $k$-stack obtained from $s^k$ by placing $s^{k-1}$ at its end;
the convention is that $s^2:s^1:s^0=s^2:(s^1:s^0)$.
We say that an $n$-stack $s$ is \emph{well formed} if every $k$-stack in $s$ is nonempty, for $1\leq k\leq n$ (including the whole $s$).
In configurations of our automata we allow only well formed stacks.

An $n$-th level \emph{deterministic higher-order pushdown automaton with data} ($n$-HOPAD for short) is a tuple $\Aa=(A,\Gamma,\gamma_I,Q,q_I,F,\Delta)$ 
where $A$ is an input alphabet, $\Gamma$ is a stack alphabet, $\gamma_I\in\Gamma$ is an initial stack symbol, 
$Q$ is a set of states, $q_I\in Q$ is an initial state, $F\subseteq Q$ is a set of accepting states,
and $\Delta\subseteq Q\times \Gamma\times(A\cup\{\varepsilon\})\times Q\times OP^n_\Gamma$ is a transition relation, where $OP^n_\Gamma$ contains the following operations:
\begin{itemize}
\item	$\mathsf{pop}^k$, where $1\leq k\leq n$, and 
\item	$\mathsf{push}^k_\alpha$, where $1\leq k\leq n$, and $\alpha\in\Gamma$.
\end{itemize}
The transition relation $\Delta$ should be deterministic, which means that it is partial function from $Q\times \Gamma\times(A\cup\{\varepsilon\})$ to $Q\times OP^n_\Gamma$,
which moreover is never simultaneously defined for some $(q,\alpha,a)\in Q\times\Gamma\times A$ and for $(q,\alpha,\varepsilon)$.
The letter $n$ is used exclusively for the level of pushdown automata.

A \emph{configuration} of $\Aa$ consists of a state and of a well formed stack, 
i.e.~is an element of $Q\times\Gamma_*^n$ in which the $n$-stack is well formed.
The \emph{initial} configuration consists of the initial state $q_I$ and of the $n$-stack containing only one $0$-stack, 
which is $(\gamma_I,\nodata)$.

Next, we define when $c\vdash^{(a,d)} c'$, where $c,c'$ are configurations, and $(a,d)\in (A\times\DV)\cup\{(\varepsilon,\nodata)\}$.
Let $p$ be the state of $c$, $q$ the state of $c'$, and let $(\alpha,d')$ be the topmost $0$-stack of $c$.
We have $c\vdash^{(a,d)} c'$ if
\begin{itemize}
\item	$(p,\alpha,a,q,\pop k)\in\Delta$, and the stack of $c'$ is obtained from the stack of $c$ by replacing its topmost $k$-stack $s^k:s^{k-1}$ by $s^k$ (i.e.~we remove the topmost $(k-1)$-stack),
	and either $d=\nodata$ or $d=d'$, or
\item	$(p,\alpha,a,q,\push k_\beta)$, and the stack of $c'$ is obtained from the stack of $c$ by replacing its topmost $k$-stack $s^k:s^{k-1}$ by $(s^k:s^{k-1}):s^{k-1}$,
	and then by replacing its topmost $0$-stack by $(\beta,d)$
	(i.e.~we copy the topmost $(k-1)$-stack, and then we change the topmost symbol in the copy).\footnote{
	 	In the classical definition the topmost symbol can be changed only when $k=1$ (for $k\geq 2$ it has to be $s^0=t^0$).
	 	We make this extension to have an unified definition of $\mathsf{push}^k$ for every $k$.
	 	However it is easy to simulate an automaton using one definition by an automaton using the other in such a way that the same language is recognized.
	 }
\end{itemize}
Notice that we require that the stack of $d$ is well-formed, so a $\pop k$ operation cannot be executed if the topmost $k$-stack has size $1$.
Observe that every $\push{}$ operation sets the data value in the topmost $0$-stack to the just read data value, or to $\nodata$ if performing an $\varepsilon$-transition.\footnote{
	In nondeterministic automata we should allow also $\varepsilon$-transitions performing a $\push{}$ operation which put an arbitrary (guessed) data value in the topmost $0$-stack. 
	This is of course forbidden in deterministic automata. 
}
On the other hand a $\pop{}$ operation can read a data value only if it is the same as the data value in the topmost $0$-stack;
however $0$-stacks with data values can be also removed without reading anything.

A \emph{run} $R$, from $c_0$ to $c_m$, is just a sequence $c_0\vdash^{x_1}c_1\vdash^{x_2}\dots\vdash^{x_m}c_m$. 
We denote $R(i):=c_i$ and call $|R|:=m$ the \emph{length of $R$}. 
The \emph{subrun} $\subrun{R}{i}{j}$ is $c_i\vdash^{x_{i+1}}c_{i+1}\vdash^{x_{i+2}}\dots\vdash^{x_j}c_j$.
If a run $R$ ends in the first configuration of a run $S$, we denote by $R\circ S$ the \emph{composition} of $R$ and $S$ which is defined as expected. 
The data word \emph{read} by $R$ is the data word over $A$ obtained from $x_1x_2\dots x_m$ by dropping all appearances of $(\varepsilon,\nodata)$.
A data word $w$ is \emph{accepted} by $\Aa$ if it is read by some run from the initial configuration to a configuration having an accepting state.

\paragraph*{Collapsible $n$-HOPAD.}
In Section \ref{sec:lang} we also use deterministic collapsible pushdown automata ($n$-CPAD for short).
Such automata are defined like $n$-HOPAD, with the following differences.
A $0$-stack contains now three parts: a symbol from $\Gamma$, a data value from $\DV$, and $n$ natural numbers.
The natural numbers are not used to determine which transition can be performed from a configuration.
Every $\push{}$ operation stores in the topmost $0$-stack a tuple of $n$ numbers $(k_1,k_2,\dots,k_n)$, where $k_i$ is the size of the topmost $i$-stack after this operation.
We have a new operation $\collapse i$ for $1\leq i\leq n$.
When it is performed between configurations $c$ and $c'$, then the stack of $c'$ is obtained from the stack of $c$ by removing its topmost $(i-1)$-stacks, so that only $k_i-1$ of them are left, where
$k_i$ is the $i$-th number stored in the topmost $0$-stack of $c$ (intuitively, from the topmost $i$-stack we remove all $(i-1)$-stacks on which the topmost $0$-stack is present).\footnote{
	Again, a classical definition of the links stored and of $\collapse{}$ differs slightly, but this does not change the class of recognized languages.
}

\section{The Separating Language}\lab{sec:lang}

In this section we define a data language $U$ which can be recognized by a $2$-CPAD, but not by any $n$-HOPAD, for any $n$.
It is a data language over the alphabet $A=\{[,],\$\}$.
The language $U$ contains the data words $w$ over the alphabet $A=\{[,],\$\}$ which satisfy the following conditions:
\begin{itemize}
\item	$w$ contains exactly one dollar, and
\item	$w$ after removing the dollar is a well-formed bracket expression, i.e.~in each prefix of $w$ the number of closing brackets is not greater than the number of opening brackets, and in the whole $w$ these numbers are equal, and
\item	the suffix of $w$ after the dollar is symmetric to some prefix of $w$ (i.e.~the $i$-th letter from the beginning and the $i$-th letter from the end have the same data value, and one of them is $]$ and the other $[$).
\end{itemize}
Notice that words in the language can be divided into three parts, the last of them preceded by a dollar.
The first part is completely symmetric to the third part, while the second part is an arbitrary bracket expression.
The first part ends on the last opening bracket which is not closed before the dollar (or this part is empty if all brackets are closed).

Let us now show briefly how $U$ can be recognized by a $2$-CPAD.
The $2$-CPAD will use 4 stack symbols: $[$, $]$, $X$ (used to mark the bottom of $1$-stacks), $Y$ (used to count open brackets).
The initial symbol is $X$.
The automaton first performs $\mathsf{push}^2(X)$.
Then, for every bracket $a$ on the input,
\begin{itemize}
\item	we perform $\push1(a)$ (we place the bracket on the stack, together with the data value),
\item	we perform $\push2$ and $\pop1$,
\item	if $a=[$ we perform $\push1(Y)$, otherwise $\pop 1$.
\end{itemize}
When at some moment the topmost symbol is $X$ and a closing bracket is read, the word is rejected (the number of closing brackets till now is greater than the number of opening brackets).
After reading a word of length $k$, we have $k+2$ $1$-stacks.
For $1\leq i\leq k$, the topmost symbol of the $(i+1)$-th $1$-stack is equal to the $i$-th symbol on the input (including the data value);
additionally the number of the $Y$ symbols on the $(i+2)$-th $1$-stack is equal to the number of opening brackets minus the number of closing brackets among the first $i$ letters of the input.

Finally the dollar is read.
If the topmost symbol is $X$, we have read as many opening brackets as closing brackets, hence we should accept immediately.
Otherwise, the topmost $Y$ symbol corresponds to the last opening bracket which is not closed.
We execute the collapse operation.
It removes the $1$-stacks containing all the brackets read after this bracket.
Now we repeat $\pop2$ until only the first $1$-stack is left (which has $X$ as its topmost symbol),
and during each $\pop2$ we read appropriate bracket from the input (the automaton ensures that the data value on the input is the same as the data value on the stack, hence the same as in the first part of the input).

In the remaining part of the paper we prove that any $n$-HOPAD cannot recognize $U$;
in particular all automata appearing in the following sections does not use collapse.

\paragraph*{A proof for level $1$.}
In order to show intuitions how the proof of Theorem \ref{thm:main} works, we present it first for level $1$.
Thus consider a $1$-HOPAD $\Aa$ which recognizes $U$.
We give him a data word $w$ with $\pi(w)=[^N]^N$ and with every data value different, where $N$ is a big enough number (namely, a number greater than the number of states).

Notice that the unique run of $\Aa$ on $w$ ($\Aa$ is deterministic) stores all data values of $w$ on the stack, and none of them can be removed from the stack until the end of $w$ is reached.
Indeed, we can afterwards give $[\$$ to the automaton, and he has to reproduce all data values of $w$ (the only suffix after which $\Aa$ should accept contains all the data values of $w$).

Let $s$ be the stack after reading the part $[^N$.
Now assume that a dollar appear after some number $k$ of closing brackets (i.e.~we have $[^N]^k\$$ instead of $[^N]^N$).
Just after the dollar, the stack of $\Aa$ still begins with $s$.
As the data values from the $[^N$ part appear only in $s$, we have to reach a configuration $(q_k,s)$ for some state $q_k$.
As we have not enough states, for two numbers of closing brackets $k, k'$, the state will be the same, i.e.~$q_k=q_{k'}$.
Afterwards, the unique accepting run from $(q_k,s)$ and from $(q_{k'},s)$ is the same, in particular the same subset of data values appearing in $s$ is read.
This is not what is expected by the language $U$, as in one case the first $N-k$ data values of $w$ should be read, and in the other case the first $N-k'$ of them.
We have a contradiction.

For higher levels the overall idea will be the same.
In the next section we will define upper runs and returns.
In our level-$1$ example, the part of the run which does not modify $s$ will be an upper run, and the part which analyzes $s$ will correspond to a return.
Next we will say that a stack can be accessed only in a constant number of ways (like here we could only visit it for each of the states).
These ways will be described by ``run descriptors'' in Section \ref{sec:types};
for each of these ways we can recover only one subset of the data values stored in the stack.
Finally, in Section \ref{sec:final}, we conclude the proof.

\section{Upper runs and returns}

In this section we define two important classes of runs, called upper runs and returns,
and we show how they decompose.

A run $R$ is called \emph{$k$-upper}, where $0\leq k\leq n$, if
the topmost $k$-stack of $R(|R|)$ is a copy of the topmost $k$-stack of $R(0)$, but possibly some changes were made to it.
This has to be properly understood for $k=0$: we assume that a $\push 0$ operation makes a copy of the topmost $0$-stack, and then modifies its contents; 
in particular a run which just performs $\push 0$ is $0$-upper.
Notice that every run is $n$-upper, and that a $k$-upper run is also $l$-upper for $l\geq k$.

A run $R$ is called a \emph{$k$-return}, where $1\leq k\leq n$, if 
\begin{itemize}
\item the topmost $(k-1)$-stack of $R(|R|)$ is obtained as a copy of the second topmost $(k-1)$-stack of $R(0)$ (in particular we require that the topmost $k$-stack has size at least $2$), and
\item while tracing this copy of the second topmost $(k-1)$-stack of $R(0)$ which finally becomes the topmost $(k-1)$-stack of $R(|R|)$, it is never the topmost $(k-1)$-stack of $R(i)$ for all $i<|R|$.
\end{itemize}
Notice that a run $R$ is a $k$-return exactly when the topmost $k$-stack of $R(|R|)$ is obtained from the topmost $k$-stack of $R(0)$ by removing its topmost $(k-1)$-stack,
and this removing was done in the last step of $R$.
In particular every $k$-return is $k$-upper.

\begin{table*}
\caption{Stack contents of the example run, and subruns being $k$-upper runs and $k$-returns}
\label{tab:example}
$$\begin{array}{c|l|l|l|l|l}
j&\mbox{stack of }R(j)&\{i:\subrun{R}{i}{j}\ 0\mbox{-upper}\}&\{i:\subrun{R}{i}{j}\ 1\mbox{-upper}\}&\{i:\subrun{R}{i}{j}\ 1\mbox{-return}\}&\{i:\subrun{R}{i}{j}\ 2\mbox{-return}\}\\
\hline
0&[ab][cd]&\{0\}&\{0\}&\emptyset&\emptyset\\
1&[ab][cd][ce]&\{0,1\}&\{0,1\}&\emptyset&\emptyset\\
2&[ab][cd][c]&\{2\}&\{0,1,2\}&\{0,1\}&\emptyset\\
3&[ab][cd]&\{0,3\}&\{0,3\}&\emptyset&\{1,2\}\\
4&[ab][c]&\{4\}&\{0,3,4\}&\{0,3\}&\emptyset\\
5&[ab][cd]&\{4,5\}&\{0,3,4,5\}&\emptyset&\emptyset\\
6&[ab][c]&\{4,6\}&\{0,3,4,5,6\}&\{5\}&\emptyset
\end{array}$$
\end{table*}

\begin{example}
Consider a HOPAD of level 2.
In this example we omit data values; brackets are used to group symbols in one $1$-stack.
Consider a run $R$ of length $6$ which begins with a stack $[ab][cd]$, and the operations between consecutive configurations are:
$$\mathsf{push}^2(e),\ \mathsf{pop}^1,\ \mathsf{pop}^2,\ \mathsf{pop}^1,\ \mathsf{push}^1(d),\ \mathsf{pop}^1.$$
Recall that our definition is that a $\mathsf{push}$ of any level can change the topmost stack symbol.
The contents of the stacks of the configurations in the run, and subruns being $k$-upper runs and $k$-returns are presented in Table \ref{tab:example}.
Notice that $R$ is not a $1$-return.
\end{example}

It is important that returns and upper runs can be always divided into smaller fragments which are again returns and upper runs.
More precisely, from \cite{parys-panic-new} we have the following characterizations of $k$-returns and $k$-upper runs (these decompositions can be obtained quite easily by analyzing all possible behaviors of an automaton).

\begin{prop}\lab{prop:return}
	A run $R$ is an $r$-return (where $1\leq r\leq n$) if and only if
	\begin{enumerate}
	\item	$|R|=1$, and the operation performed by $R$ is $\mathsf{pop}^r$, or
	\item	the first operation performed by $R$ is $\mathsf{pop}^k$ for $k<r$, or $\mathsf{push}^k$ for $k\neq r$, and $\subrun{R}{1}{|R|}$ is an $r$-return, or
	\item	the first operation performed by $R$ is $\mathsf{push}^k$ for $k\geq r$, and $\subrun{R}{1}{|R|}$ is a composition of a $k$-return and an $r$-return.
	\end{enumerate}
\end{prop}

\begin{prop}\lab{prop:return-bis}
	A run $R$ is $k$-upper (where $0\leq k\leq n$) if and only if
	\begin{enumerate}
	\item	$R$ performs only operations of level at most $k$, or
	\item	$|R|=1$, and the operation performed by $R$ is $\mathsf{push}^r$ for any $r\geq k+1$, or
	\item	the first operation performed by $R$ is $\mathsf{push}^r$ for $r\geq k+1$, and $\subrun{R}{1}{|R|}$ is an $r$-return, or
	\item	$R$ is a composition of two nonempty $k$-upper runs.
	\end{enumerate}
\end{prop}

\section{Types of Stacks}\lab{sec:types}

In this section we assign to each $k$-stack a type from a finite set which, in some sense, describes possible returns and upper runs from a configuration having this $k$-stack as the topmost $k$-stack.
Additionally, we also define a set of data values which can be extracted from a $k$-stack; the idea is that we can extract either none or all data values contained in such a set.

For this section we fix an $n$-HOPAD $\Aa$ with stack alphabet $\Gamma$, state set $Q$, and input alphabet $A$.
Moreover we fix a morphism $\varphi\colon A^*\to M$ into a finite monoid.
In our application the monoid will be testing whether the word is empty, or is a single closing bracket, or begins with a dollar, or is any other word; however the results of this section work for any monoid.
For a run $R$ we write $\varphi(R)$ to denote the value of $\varphi$ on the word read by $R$ 
(more precisely $\varphi$ is applied only to the first coordinates of the letters of this word).

\subsection{Returns}

We begin by defining types which describe returns.
In the next subsection we indirectly use the same types to describe upper runs;
this will be possible because upper runs consist of returns (as described by Proposition \ref{prop:return-bis}).
To every $k$-stack $s^k$ (where $0\leq k\leq n$) we assign a set $\type(s^k)\subseteq\Tt^k$; it contains some \emph{run descriptors}.
The sets $\Tt^k$ of all possible run descriptors are defined inductively as follows (where $\Pp(X)$ denotes the power set of $X$):
\begin{align*}
	\Tt^k &= \{\mathsf{ne}\}\cup\big(\Pp(\Tt^n)\times\Pp(\Tt^{n-1})\times\dots\times\Pp(\Tt^{k+1})\times Q\times \Dd^k\big),\quad\mbox{where}\\
	\Dd^k &= \bigcup_{r=k+1}^nM\times\{r\}\times\Pp(\Tt^n)\times\Pp(\Tt^{n-1})\times\dots\times\Pp(\Tt^{r+1})\times Q.
\end{align*}
Recall that $Q$ is the set of states, and $M$ is the finite monoid used to distinguish classes of input words.

Before we define types, let us describe their intended meaning.
We say that a run $R$ \emph{agrees} with $(m,r,\Sigma^n,\Sigma^{n-1},\dots,\Sigma^{r+1},q)\in\Dd^k$ if
\begin{itemize}
\item	$\varphi(R)=m$, and
\item	$R$ is an $r$-return, and
\item	$R$ ends in a stack $t^n:t^{n-1}:\dots:t^r$ such that $\Sigma^i\subseteq \type(t^i)$ for $r+1\leq i\leq n$, and
\item	the state of $R(|R|)$ is $q$.
\end{itemize}
The acronym $\mathsf{ne}$ stands for nonempty; we have $\mathsf{ne}\in \type(s^k)$ when $s^k$ is nonempty.
A typical run descriptor in $\Tt^k$ is of the form $\sigma=(\Psi^n,\Psi^{n-1},\dots,\Psi^{k+1},p,\widehat\sigma)$.
By adding $\sigma$ to the type of some $s^k$, we claim the following.
If for each $k+1\leq i\leq n$ we take an $i$-stack $t^i$ that satisfies the claims of $\Psi^i$,
then there is a run which starts in state $p$ and stack $t^n:t^{n-1}:\dots:t^{k+1}:s^k$, and agrees with $\widehat\sigma$.
In other words we have the following lemma.
\begin{lemma}\lab{lem:run2type}
	Let $0\leq k\leq n$, let $\widehat\sigma\in\Dd^k$, and let $c=(p,s^n:s^{n-1}:\dots:s^k)$ be a configuration.
	The following two conditions are equivalent:
	\begin{enumerate}
	\item	there exists a run from $c$ which agrees with $\widehat\sigma$, 
	\item	$\type(s^k)$ contains a run descriptor $(\Psi^n,\Psi^{n-1},\dots,\Psi^{k+1},p,\widehat\sigma)$
		such that $\Psi^i\subseteq \type(s^i)$ for $k+1\leq i\leq n$.
	\end{enumerate}
\end{lemma}

Moreover, to every $k$-stack $s^k$ (where $0\leq k\leq n$), and every run descriptor $\sigma\in\type(s^k)$ 
we assign a set of \emph{important data values} $\idv(s^k,\sigma)\subseteq\DV$.
The idea is that although $s^k$ may contain many data values,
a run having $s^k$ in its first configuration and described by $\sigma$ can use only some subset of these data values; 
these will be the important data values of $s^k$ (from the point of view of $\sigma$).
There is a small technical difficulty that a run can read (by accident) a data value already present in $s^k$ while performing a $\push{}$ operation;
of course this is not what we want to count.
To deal with this, we distinguish some data value $0\in\DV$ and we define normalized runs:
a run $R$ is \emph{normalized} when every of its non-$\varepsilon$-transition performing a $\push{}$ reads data value $0$.
The sets $\idv$ are designed to satisfy the following lemma.

\begin{lemma}\lab{lem:idv}
	Let $0\leq k\leq n$, let $\widehat\sigma=(m,r,\Sigma^n,\Sigma^{n-1},\dots,\Sigma^{r+1},q)\in\Dd^k$, let $c=(p,s^n:s^{n-1}:\dots:s^k)$ be a configuration, 
	and let $d\in\DV\setminus\{0\}$.
	The following two conditions are equivalent: 
	\begin{enumerate}
	\item	there exists a normalized $r$-return from $c$ which agrees with $\widehat\sigma$, and either reads $d$ or ends in a stack $t^n:t^{n-1}:\dots:t^r$ such that 
		$d\in\idv(t^i,\tau)$ for some $r+1\leq i\leq n$ and some $\tau\in\Sigma^i$,
	\item	$\type(s^k)$ contains a run descriptor $\sigma=(\Psi^n,\Psi^{n-1},\dots,\Psi^{k+1},p,\widehat\sigma)$ such that $\Psi^i\subseteq \type(s^i)$ for $k+1\leq i\leq n$,
		and $d\in\idv(s^k,\sigma)$ or $d\in\idv(s^i,\tau)$ for some $k+1\leq i\leq n$ and some $\tau\in\Psi^i$.
	\end{enumerate}
\end{lemma}

In the first condition we only say that there exists such run, moreover it can be different for every data value.
However we will use this lemma only in situations when this run is in some sense unique, so that this unique run will be described.
Observe also that a data value which was important in the first configuration of a run need not to be always read: it may also stay important in the last configuration of the run.
We can exclude this possibility by taking empty sets $\Sigma^i$ (in particular by taking $r=n$).

Now we come to the definition of types and important data values.
We first define a composer, which is then (in Definition \ref{def:types-k}) used to compose types of smaller stacks into types of greater stacks.

\begin{definition}\lab{def:composer}
	We define when $(\Phi^k,\Phi^{k-1}\dots,\Phi^l;\Psi^k)$ is a \emph{composer}, where $0\leq l<k\leq n$, 
	$\Phi^i\subseteq\Tt^i$ for $l\leq i\leq k$, and $\Psi^k\subseteq\Tt^k$.
	\begin{enumerate}
	\item\lab{pkt:composer1}
		Let $\sigma=(\Sigma^n,\Sigma^{n-1},\dots,\Sigma^{l+1},p,\widehat\sigma)\in\Tt^l$, and let
		$\tau=(\Sigma^n,\Sigma^{n-1},\dots,\Sigma^{k+1},p,\widehat\sigma)\in\Tt^k$ (in particular $\widehat\sigma\in\Dd^k$).
		Then we say that $(\Sigma^k,\Sigma^{k-1},\dots,\Sigma^{l+1},\{\sigma\};\{\tau\})$ is a composer.
	\item
		Tuple $(\emptyset,\emptyset,\dots,\emptyset;\{\mathsf{ne}\})$ is a composer.
	\item
		Let $\Psi^k\subseteq\Tt^k$, and for each $\tau\in\Psi^k$ let $(\Phi_\tau^k,\Phi_\tau^{k-1},\dots,\Phi_\tau^l;\{\tau\})$ be a composer.
		Then we say that $(\Phi^k,\Phi^{k-1},\dots,\Phi^l;\Psi^k)$ is a composer, where $\Phi^i=\bigcup_{\tau\in\Psi^k}\Phi_\tau^i$ for $l\leq i\leq k$.
	\end{enumerate}
\end{definition}

Now we will define types of $0$-stacks, and sets of important data values as a fixpoint.

\begin{definition}\lab{def:types}
	For each $0$-stack $s^0$, let $\type(s^0)$ be the smallest set satisfying the following conditions;
	additionally for $\sigma\in\type(s^0)$, let $\idv(s^0,\sigma)$ be the smallest set satisfying the following conditions.
	Let $s^0=(\alpha,d)$ be a $0$-stack, and let $p$ be a state.
	\begin{enumerate}
	\item
		Assume that $\Aa$ has a transition $(\alpha,p,a,q_1,\pop k)$ (where $1\leq k\leq n$ and $a\in A\cup\{\varepsilon\}$), and that either $a=\varepsilon$ or $d\not=\nodata$.
		We have two subcases:
		\begin{enumerate}
		\item\lab{pkt:types1a}
			Let 
			\begin{align*}
				\sigma=\left(\Psi^n,\Psi^{n-1},\dots,\Psi^{k+1},\{\mathsf{ne}\},\emptyset,\emptyset,\dots,\emptyset,p,(\varphi(a),k,
				\Psi^n,\Psi^{n-1},\dots,\Psi^{k+1},q_1)\right)\in\Tt^0.
			\end{align*}
			Then $\sigma\in\type(s^0)$. If $a\not=\varepsilon$, then $d\in\idv(s^0,\sigma)$.
		\item\lab{pkt:types1b}
			Let $\tau=(\Psi^n,\Psi^{n-1},\dots,\Psi^{k+1},q_1,(m,r,\Sigma^n,\Sigma^{n-1},\dots,\Sigma^{r+1},q))\in \Tt^k$, and
			$$\sigma=(\Psi^n,\Psi^{n-1},\dots,\Psi^{k+1},\{\tau\},\emptyset,\emptyset,\dots,\emptyset,p,(\varphi(a)m,r,\Sigma^n,\Sigma^{n-1},\dots,\Sigma^{r+1},q))\in \Tt^0.$$
			Then $\sigma\in\type(s^0)$. If $a\not=\varepsilon$, then $d\in\idv(s^0,\sigma)$.
		\end{enumerate}
	\item
		Assume that $\Aa$ has a transition $(\alpha,p,a,q_1,\push k(\beta))$ (where $1\leq k\leq n$).
		Let $t^0=(\beta,\nodata)$ if $a=\varepsilon$, or $t^0=(\beta,0)$ otherwise.
		Assume that $\type(t^0)$ contains a run descriptor
		$$\tau=(\Psi^n,\Psi^{n-1},\dots,\Psi^1,q_1,(m_1,r_1,\Omega^n,\Omega^{n-1},\dots,\Omega^{r_1+1},q_2)).$$
		Assume also that $(\Phi^k,\Phi^{k-1},\dots,\Phi^0;\Psi^k)$ is a composer such that $\Phi^0\subseteq \type(s^0)$.
		We have two subcases:
		\begin{enumerate}
		\item\lab{pkt:types2a}
			Assume that $r_1\neq k$.
			Let
			$$\sigma=(\Pi^n,\Pi^{n-1},\dots,\Pi^1,p,(\varphi(a)m_1,r_1,\Omega^n,\Omega^{n-1},\dots,\Omega^{r_1+1},q_2))\in \Tt^0,$$
			where
			$$\Pi^i=\left\{\begin{array}{ll}
				\Psi^i  & \mbox{for }k+1\leq i\leq n,\\
				\Phi^i  & \mbox{for }i=k,\\
				\Psi^i\cup\Phi^i  & \mbox{for }1\leq i\leq k-1.
			\end{array}\right.$$
			Then $\sigma\in\type(s^0)$, and $\bigcup_{\rho\in\Phi^0}\idv(s^0,\rho)\subseteq\idv(s^0,\sigma)$.
		\item\lab{pkt:types2b}
			Assume that $r_1=k$ and $\type(s^0)$ contains a run descriptor
			\begin{align*}\chi=\left(\Omega^n,\Omega^{n-1}\dots,\Omega^{k+1},\Upsilon^k,\Upsilon^{k-1},\dots,\Upsilon^1,q_2,(m_2,r_2,
			\Sigma^n,\Sigma^{n-1},\dots,\Sigma^{r_2+1},q_3)\right),\end{align*}
			where $r_2\leq k$. 
			Let
			$$\sigma=(\Pi^n,\Pi^{n-1},\dots,\Pi^1,p,(\varphi(a)m_1m_2,r_2,\Sigma^n,\Sigma^{n-1},\dots,\Sigma^{r_2+1},q_3))\in \Tt^0,$$
			where
			$$\Pi^i=\left\{\begin{array}{ll}
				\Psi^i  & \mbox{for }k+1\leq i\leq n,\\
				\Phi^i\cup\Upsilon^i  & \mbox{for }i=k,\\
				\Psi^i\cup\Phi^i\cup\Upsilon^i  & \mbox{for }1\leq i\leq k-1.
			\end{array}\right.$$
			Then $\sigma\in\type(s^0)$, and $\idv(s^0,\chi)\cup\bigcup_{\rho\in\Phi^0}\idv(s^0,\rho)\subseteq\idv(s^0,\sigma)$.
		\end{enumerate}
	\end{enumerate}
\end{definition}

Notice that the above operations are monotone (more run descriptors in types can only cause that more run descriptors are added to other types),
thus the unique smallest fixpoint exists.
Finally, we define types of greater stacks.

\begin{definition}\lab{def:types-k}
	For each $k$-stack $s^k$, where $1\leq k\leq n$, let $\type(s^k)$ be the smallest set satisfying the following conditions;
	additionally for $\sigma\in\type(s^k)$, let $\idv(s^k,\sigma)$ be the smallest set satisfying the following conditions.
	\begin{quote}
		Assume that $s^k=t^k:t^{k-1}$, and there exists a composer $(\Phi^k,\Phi^{k-1};\{\sigma\})$ for which $\Phi^k\subseteq\type(t^k)$ and $\Phi^{k-1}\subseteq\type(t^{k-1})$.
		Then $\sigma\in\type(s^k)$, and $\bigcup_{i=k-1}^k\bigcup_{\rho\in\Phi^i}\idv(t^i,\rho)\subseteq\idv(s^k,\sigma)$.
	\end{quote}	
\end{definition}

Let us now observe immediate consequences of Definitions \ref{def:composer} and \ref{def:types-k}.
We see that $\mathsf{ne}$ is in the type of a $k$-stack $s^k$ if and only if $k\geq 1$ and $s^k$ is nonempty.
We also see that a run descriptor $\sigma$ is in the type of a $k$-stack $s^k:s^{k-1}:\dots:s^l$ (where $l<k$) 
if and only if there exists a composer $(\Phi^k,\Phi^{k-1},\dots,\Phi^l;\{\sigma\})$
such that $\Phi^i\subseteq\type(s^i)$ for $l\leq i\leq k$ (moreover $\Phi^l$ has exactly one element).
Similarly, $d\in\idv(s^k:s^{k-1}:\dots:s^l, \sigma)$ if and only if there exists a composer $(\Phi^k,\Phi^{k-1},\dots,\Phi^l;\{\sigma\})$
such that $\Phi^i\subseteq\type(s^i)$ for $l\leq i\leq k$ and that $d\in\idv(s^j,\tau)$ for some $l\leq j\leq k$ and some $\tau\in\Phi^j$.

It is tedious but straightforward to prove Lemmas \ref{lem:run2type} and \ref{lem:idv} directly from the definitions of $\type$ and $\idv$.
For the 1$\Rightarrow$2 implications, we decompose the run according to Proposition \ref{prop:return}; 
this tells us which rules of Definition \ref{def:types} imply that an appropriate run descriptor is in the type, and that the considered data value in in $\idv$.
For the opposite implication, a run descriptor is in the type (or a data value is in $\idv$), because there is some derivation using the rules from Definition \ref{def:types}; 
by composing all these rules we obtain an appropriate run.

Let us mention that a similar notion of types was also present in \cite{parys-pumping}.
Those types were defined in a different, semantical way.
Namely, our Lemma \ref{lem:run2type} is used as a definition;
then it is necessary to prove that the type of $s^k$ does not depend on the choice of $s^n, s^{n-1},\dots,s^{k+1}$ present in the assumptions of the lemma.
This semantical definition is unsuitable for tracing important data values.
Our approach is better for the following reason:
when we add some run descriptor $(\Psi^n,\Psi^{n-1},\dots,\Psi^{k+1},p,\widehat\sigma)$ to the type of some $k$-stack, 
then in the sets $\Psi^i$ we only have the assumptions which are really useful (we never put there redundant run descriptors).

\subsection{Upper runs}

For upper runs we do not define separate types; we will be using here the types for returns.
Our goal is to obtain the following transfer property.

\begin{lemma}\lab{lem:idv-upper}
	Let $0\leq k\leq n$, and let $R$ be a normalized $k$-upper run such that no other normalized run $S$ from $R(0)$ ends in the same state as $R$ and satisfies $\varphi(S)=\varphi(R)$.
	Let $s^n:s^{n-1}:\dots:s^k$ be the stack or $R(0)$, and $t^n:t^{n-1}:\dots:t^k$ the stack of $R(|R|)$.
	Let $d,d'\in\DV\setminus\{0\}$ be data values which are not read by $R$ and which do not appear in $s^k$, and such that for each $k+1\leq i\leq n$ and each $\sigma\in\type(s^i)$
	it holds $d\in\idv(s^i,\sigma)\Iff d'\in\idv(s^i,\sigma)$.
	Then $d$ and $d'$ do not appear in $t^k$, and for each $k+1\leq i\leq n$ and each $\sigma\in\type(t^i)$ it holds $d\in\idv(t^i,\sigma)\Iff d'\in\idv(t^i,\sigma)$.
\end{lemma}

Intuitively, the lemma says that if two data values cannot be distinguished by the $\idv$ sets in the initial configuration of a $k$-upper run, they also cannot be distinguished in the final configuration.
In order to obtain this lemma we define $\origin$ sets as follows (``$\origin$'' stands for ``source'').
Let $0\leq k\leq n$, let $R$ be a normalized $k$-upper run from stack $s^n:s^{n-1}:\dots:s^k$ to stack $t^n:t^{n-1}:\dots:t^k$,
and let $\Sigma^i\subseteq\type(t^i)$ for $k+1\leq i\leq n$.
For $k+1\leq j\leq n$ we define a set $\origin^j(R,\Sigma^n,\Sigma^{n-1},\dots,\Sigma^{k+1})\subseteq\type(s^j)$.
To shorten the notation, we write $\origin^j$ for $\origin^j(R,\Sigma^n,\Sigma^{n-1},\dots,\Sigma^{k+1})$ (we should keep in mind that $\origin^j$ depends on $R$ and on $\Sigma^n,\Sigma^{n-1},\dots,\Sigma^{k+1}$).
The intuition is that $d\in\idv(t^i,\sigma)$ for some $k+1\leq i\leq n$ and some $\sigma\in\Sigma^i$ if and only if $d\in\idv(s^j,\tau)$ for some $k+1\leq j\leq n$ and some $\tau\in\origin^j$.
The actual statement is more complicated, and is given by the following lemma.

\begin{lemma}\lab{lem:origin}
	Let $0\leq k\leq n$, let $R$ be a normalized $k$-upper run from $(p,s^n:s^{n-1}:\dots:s^k)$ to $(q,t^n:t^{n-1}:\dots:t^k)$,
	and let $\Sigma^i\subseteq\type(t^i)$ for $k+1\leq i\leq n$.
	Let also $d\in\DV\setminus\{0\}$ be a data value which does not appear in $s^k$.
	\begin{enumerate}
	\item\lab{pkt:origin1}
		Assume that $d\in\idv(t^i,\sigma)$ for some $k+1\leq i\leq n$ and some $\sigma\in\Sigma^i$.
		Then $d\in\idv(s^j,\tau)$ for some $k+1\leq j\leq n$ and some $\tau\in\origin^j$.
	\item\lab{pkt:origin2}
		Let $c=(p,u^n:u^{n-1}:\dots:u^k)$ be a configuration such that $u^k=s^k$ and $\origin^i\subseteq\type(u^i)$ for $k+1\leq i\leq n$.
		Assume that $d\in\idv(u^j,\tau)$ for some $k+1\leq j\leq n$ and some $\tau\in\origin^j(R,\Sigma^n,\Sigma^{n-1},\dots,\Sigma^{k+1})$.
		Then there exists a normalized $k$-upper run $S$ from $c$ to a configuration $(q,v^n:v^{n-1}:\dots:v^k)$ such that 
		$\varphi(S)=\varphi(R)$,
		and $v^k=t^k$,
		and $\Sigma^i\subseteq\type(v^i)$ for $k+1\leq i\leq n$,
		and either $S$ reads $d$ or 
		$d\in\idv(v^i,\sigma)$ for some $k+1\leq i\leq n$ and some $\sigma\in\Sigma^i$.
	\end{enumerate}
\end{lemma}

Now we come to the definition of the $\origin$ sets.
We define them by induction on the length of $R$,
and we make a case distinction on the form of $R$ as in Proposition \ref{prop:return-bis}.
\begin{enumerate}
\item
	Assume that $R$ performs only operations of level at most $k$.
	Then we take $\origin^i=\Sigma^i$ for $k+1\leq i\leq n$.
\item
	Assume that $|R|=1$ and the operation performed by $R$ is $\mathsf{push}^r$ for $r\geq k+1$.
	Then we define $\origin^i$ as the smallest sets satisfying the following:
	\begin{itemize}
	\item	$\Sigma^i\subseteq\origin^i$ for $k+1\leq i\leq n$, $i\not=r$, and
	\item	for every composer $(\Psi^r,\Psi^{r-1},\dots,\Psi^k;\{\sigma\})$ such that $\sigma\in\Sigma^r$ and $\Psi^i\subseteq\type(s^i)$ for $k\leq i\leq r$ it holds $\Psi^i\subseteq\origin^i$ for $k+1\leq i\leq r$.
	\end{itemize}
\item
	Assume that the first operation of $R$ is $\mathsf{push}^r$, and $\subrun{R}{1}{|R|}$ is an $r$-return, where $r\geq k+1$.
	Let $\widehat\rho=(\varphi(\subrun{R}{1}{|R|}),r,\Sigma^n,\Sigma^{n-1},\dots,\Sigma^{r+1},q)$, where $q$ is the state of $R(|R|)$, and let $\overline s^k$ be the topmost
	$k$-stack of $R(1)$.
	We define $\origin^i$ as the smallest sets satisfying the following:
	\begin{itemize}
	\item	$\Sigma^i\subseteq\origin^i$ for $k+1\leq i\leq r$, and
	\item	for every run descriptor $\rho=(\Phi^n,\Phi^{n-1},\dots,\Phi^{k+1},q_1,\widehat\rho)\in\type(\overline s^k)$ such that $\Phi^i\subseteq\type(s^i)$ for $k+1\leq i\leq n$, $i\not=r$, and $\Phi^r\subseteq\type(s^r:s^{r-1}:\dots:s^k)$,
		it holds $\Phi^i\subseteq\origin^i$ for $k+1\leq i\leq n$, $i\not=r$, and
	\item	for every run descriptor $\rho=(\Phi^n,\Phi^{n-1},\dots,\Phi^{k+1},q_1,\widehat\rho)\in\type(\overline s^k)$ such that $\Phi^i\subseteq\type(s^i)$ for $k+1\leq i\leq n$, $i\not=r$, and $\Phi^r\subseteq\type(s^r:s^{r-1}:\dots:s^k)$,
		and for every composer $(\Psi^r,\Psi^{r-1},\dots,$ $\Psi^k;\{\lambda\})$ such that $\lambda\in\Phi^r$ and $\Psi^i\subseteq\type(s^i)$ for $k\leq i\leq r$ it holds $\Psi^i\subseteq\origin^i$ for $k+1\leq i\leq r$.
	\end{itemize}
\item
	Assume that $R=R_1\circ R_2$ for shorter $k$-upper runs $R_1,R_2$ (if there are multiple choices of $R_1,R_2$, we just fix any of them), and case 1 does not hold.
	The $\origin^i$ sets for $R_1$ and for $R_2$ are already defined by induction.
	Denote $\Phi^i=\origin^i(R_2,\Sigma^n,\Sigma^{n-1},\dots,\Sigma^{k+1})$ for $k+1\leq i\leq n$.
	We define $\origin^i:=\origin^i(R_1,\Phi^n,\Phi^{n-1},\dots,\Phi^{k+1})$ for $k+1\leq i\leq n$.
\end{enumerate}

The proof of Lemma \ref{lem:origin} is tedious but again follows easily from the definition.
In the third case we use Lemma \ref{lem:idv}.
Notice that point 1 of Lemma \ref{lem:idv} tells us only that a run exists, so we need a similar statement in point 2 of our lemma.
We finish by a proof of Lemma \ref{lem:idv-upper}.

\begin{proof}[Proof of Lemma \ref{lem:idv-upper}]
	The part that $d$ and $d'$ do not appear in $t^k$ is immediate, as by definition of a $k$-upper run, $t^k$ is a (possibly modified) copy of $s^k$.
	Because $s^k$ does not contain $d$ nor $d'$, and $R$ is normalized, they cannot appear on this stack.
	
	Assume that $d\in\idv(t^i,\sigma)$ for some $k+1\leq i\leq n$ and some $\sigma\in\type(t^i)$.
	Take $\Sigma^l=\emptyset$ for $k+1\leq l\leq n$, $l\not=i$, and $\Sigma^i=\{\sigma\}$.
	By Lemma \ref{lem:origin}.\ref{pkt:origin1} $d\in\idv(s^j,\tau)$ for some $k+1\leq j\leq n$ and some $\tau\in\origin^j$.
	By assumption also $d'\in\idv(s^j,\tau)$.
	Now we use Lemma \ref{lem:origin}.\ref{pkt:origin2} for configuration $c=R(0)$.
	We obtain a run $S$, which ends in the same state as $R$, and such that $\varphi(S)=\varphi(R)$; this means that $R=S$.
	The lemma implies that $d'\in\idv(t^l,\sigma')$ for some $k+1\leq l\leq n$ and some $\sigma'\in\Sigma^l$.
	As only $\Sigma^i$ is nonempty and contains only $\sigma$, we obtain $d'\in\idv(t^i,\sigma)$.
	The opposite implication is symmetric.
\end{proof}

\section{Why $U$ Cannot Be Recognized?}\lab{sec:final}

In this section we prove that language $U$ cannot be recognized by a deterministic higher-order pushdown automaton with data of any level.
We base on the techniques developed in the previous sections.

Of course our proof goes by contradiction: assume that for some $n$ we have an $(n-1)$-HOPAD recognizing $U$.
We construct an $n$-HOPAD $\Aa$ which works as follows.
First it makes a $\mathsf{push}^n$ operation.
Then it simulates the $(n-1)$-HOPAD (not using the $\mathsf{push}^n$ and $\mathsf{pop}^n$ operations).
When the $(n-1)$-HOPAD is going to accept, $\Aa$ makes the $\mathsf{pop}^n$ operation and afterwards accepts.
Clearly, $\Aa$ recognizes $U$ as well.
Such normalization allows us to use Lemmas \ref{lem:run2type} and \ref{lem:idv}, as in $\Aa$ every accepting run is an $n$-return.
Moreover, we assume that to each state lead either only $\varepsilon$-transitions, or only non-$\varepsilon$-transitions (it can be easily achieved by having two copies of each state).

As already mentioned, let $\varphi\colon A^*\to M$ be a morphism into a finite monoid which tests whether the word is empty, or is a single closing bracket, or begins by a dollar, or is any other word.
This morphism is used with the definition of types and important data values in the previous section.

Let $N=2+2^M$, where $M=\sum_{i=0}^n|\Tt^i|$.
Consider the following words:
\begin{align*}
	w_0=[][\qquad\qquad
	w_{k+1}=w_k^N]^N[\qquad\mbox{for }0\leq k\leq n-1,
\end{align*}
where the number in the superscript (in this case $N$) denotes the number of repetitions of a word.
Let $\overline w_n$ be a data word in which every data value is different and is not equal to $0$, and such that $\pi(\overline w_n)=w_n$.

Let $R$ be the unique run from the initial configuration which reads $\overline w_n$, and ends just after reading its last letter.
Abusing slightly the notation, for any prefixes $v,w$ of $w_k$, we define $R(v)$ to be the configuration of $R$ just after reading $|v|$ letters,
and $\Rz{v,w}$ to be the subrun of $R$ between $R(v)$ and $R(w)$.
It will be important to analyze relations between these configurations $R(v)$ for different prefixes.

Notice that every data value read by $R$ appears somewhere in a stack of $R(v)$.
This is true, because after reading $[\$$, the automaton has to check in particular whether the data values present later on the input are the same as those already read by $R$.
However this contradicts with the following key lemma (taken for $k=n$ and $u=\varepsilon$).

\begin{lemma}
	Let $0\leq k\leq n$, and let $u$ be a word such that $uw_k$ is a prefix of $w_n$.
	Then some data value read by $\Rz{u,uw_k}$ does not appear in the topmost $k$-stack of $R(uw_k)$.
\end{lemma}

\begin{proof}
	The proof is by induction on $k$.
	For $k=0$ this is obvious, as three data values are read, but the topmost $0$-stack can store only one of them.

	Let now $k\geq 1$.
	Assume first that for some nonempty prefix $v$ of $w_k$ the run $\Rz{uv,uw_k}$ is not $(k-1)$-upper.
	Let us look at the data value read by the last operation before $R(uv)$.
	This data value appears only once on the stack of $R(uv)$, in the topmost $0$-stack.
	One possibility is that $\Rz{uv,uw_k}$ is not $k$-upper.
	By definition this means that the topmost $k$-stack of $R(uw_k)$ is not a copy of the topmost $k$-stack of $R(uv)$, thus it cannot contain this data value.
	The opposite possibility is that $\Rz{uv,uw_k}$ is $k$-upper, but not $(k-1)$-upper.
	This means that the topmost $k$-stack of $R(uw_k)$ is obtained as a copy of the topmost $k$-stack of $R(uv)$,
	however the topmost $(k-1)$-stack of this $k$-stack is not a copy of the topmost $(k-1)$-stack of $R(uv)$.
	The only way how this could happen is that all copies of the topmost $(k-1)$-stack of $R(uv)$ were removed from this $k$-stack.
	Then also this $k$-stack does not contain the considered data value.
	
	For the rest of the proof assume the opposite: that for each nonempty prefix $v$ of $w_k$ the run $\Rz{uv,uw_k}$ is $(k-1)$-upper.
	We will obtain a contradiction, which will prove that this is in fact impossible.
	From this assumption we get the following property $\heartsuit$.
	\begin{quote}
		Let $v'$ be a prefix of $w_k$, and $v$ a nonempty prefix of $v'$.
		Then $\Rz{uv,uv'}$ is $(k-1)$-upper.
	\end{quote}
	Indeed, if the topmost $(k-1)$-stack of $R(uw_k)$ is a copy of the topmost $(k-1)$-stack of $R(uv')$ and of the topmost $(k-1)$-stack of $R(uv)$,
	then for sure the topmost $(k-1)$-stack of $R(uv')$ is a copy of the topmost $(k-1)$-stack of $R(uv)$.
	
	From the induction assumption (where $uw_{k-1}^{i-1}$ is taken as $u$), 
	for each $2\leq i\leq N$ there exists a data value $d_i$ read by $\Rz{uw_{k-1}^{i-1},uw_{k-1}^i}$ which does not appear in the topmost $(k-1)$-stack of $R(uw_{k-1}^i)$.
	As $\Rz{uw_{k-1}^i,uw_{k-1}^N}$ is $(k-1)$-upper (property $\heartsuit$), we know that $d_i$ does not appear in the topmost $(k-1)$-stack of $R(uw_{k-1}^N)$ as well.
	
	Denote the stack of $R(uw_{k-1}^N)$ as $s^n:s^{n-1}:\dots:s^{k-1}$.
	Since there are more possible indices $i\in\{2,3,\dots,N\}$ than subsets of $\bigcup_{i=k}^n\type(s^i)$,
	there have to exist two indices $2\leq x<y\leq N$ such that 
	for each $k\leq i\leq n$ and each $\tau\in\type(s^i)$ it holds 
	$d_x\in\idv(s^i,\tau)\Iff d_y\in\idv(s^i,\tau)$.
	Let $r=N-x+1$.
	Let us consider the unique accepting run $S$ from $R(uw_{k-1}^N]^r)$ whose first read letter is a dollar
	(this run is unique because the $U$ language determines the rest of the word after a dollar).
	Notice that the last opening bracket which is not closed before the dollar is the last bracket of the $x$-th $w_{k-1}$ after $u$.
	Thus by definition of $U$, $S$ reads $d_x$ and does not read $d_y$.
	This contradicts with the thesis of the next lemma (where $k-1$ is used as $k$).
\end{proof}

\begin{lemma}
	Let $0\leq k\leq n$, and let $u$ be a word and $r$ a number such that $u]^r$ is a prefix of $w_n$.
	Assume that $\Rz{u]^{i-1},u]^i}$ is $k$-upper for each $1\leq i\leq r$.
	Denote the stack of $R(u)$ as $s^n:s^{n-1}:\dots:s^k$.
	Let also $d,d'$ be data values which are not read by $\Rz{u,u]^r}$, which does not appear in $s^k$, and such that for each $k+1\leq i\leq n$ and each $\tau\in\type(s^i)$ it holds
        $d\in\idv(s^i,\tau)\Iff d'\in\idv(s^i,\tau)$.
        Then the unique accepting run from $R(u]^r)$ whose first read letter is a dollar, either reads both $d$ and $d'$, or none of them.
\end{lemma}

\begin{proof}
	Originally we have assumed that every data value read by $R$ is different.
	Let us now make a small twist, and assume that the data values appearing with the $r$ closing brackets after $u$ are all equal to $0$.
	This does not affect the statement of the lemma, as the changed data values does not appear in the stack of $R(u)$ nor these are $d$ or $d'$.
	Thanks to this change, the runs $\Rz{u]^{i-1},u]^i}$ are normalized for $1\leq i\leq r$.
	Notice also that $\Rz{u]^{i-1},u]^i}$ is the only normalized run from $R(u]^{i-1})$ which ends in the same state as $\Rz{u]^{i-1},u]^i}$, and reads a single closing bracket
	(recall that equality of states causes that such a run ends just after reading a letter).
	Thus we can apply Lemma \ref{lem:idv-upper} for the run $\Rz{u]^{i-1},u]^i}$ and for data values $d,d'$.
	Denote the stack of $R(u]^r)$ as $t^n:t^{n-1}:\dots:t^k$.
	After applying this lemma consecutively for $i=1,2,\dots,r$, we obtain that $d,d'$ does not appear in $t^k$, and that for each $k+1\leq i\leq n$ and each $\tau\in\type(t^i)$ it holds
	$d\in\idv(t^i,\tau)\Iff d'\in\idv(t^i,\tau)$.
	
	Now consider the unique accepting run from $R(u]^r)$ whose first read letter is a dollar.
	We see that it is normalized (it does not read any new data values),
	and that this is an $n$-return (as every accepting run of $\Aa$).
	It agrees with $\widehat\sigma=(\varphi(S),n,q)$ for some accepting state $q$.
	Lemma \ref{lem:run2type} implies that $\type(t^k)$ contains a run descriptor $\sigma=(\Phi^n,\Phi^{n-1},\dots,\Phi^{k+1},p,\widehat\sigma)$ such that $\Phi^i\subseteq\type(t^i)$ for $k+1\leq i\leq n$.
	Assume that $S$ reads one of $d,d'$, say $d$.
	Then implication 1$\Rightarrow$2 of Lemma \ref{lem:idv} says that 
	$d\in\idv(t^i,\tau)$ for some $k+1\leq i\leq n$ and some $\tau\in\Phi^i$
	(the lemma allows also that $d\in\idv(t^k,\sigma)$, but this is impossible as $d$ does not appear in $t^k$).
	However in such case also $d'\in\idv(t^i,\tau)$.
	Using now the opposite implication of Lemma \ref{lem:idv} we obtain a normalized run $S'$ from $S(0)$ which agrees with $\widehat\sigma$ and reads $d'$.
	But $S'$ is also accepting and its first letter read is a dollar, so $S'=S$.
\end{proof}

%

\bibliographystyle{eptcs} \bibliography{bib}
\end{document}